\documentclass[a4paper,fleqn,11pt]{article}

\usepackage[a4paper,top=3cm, bottom=3cm, left=3cm, right=3cm]{geometry}
\usepackage{amsmath}
\usepackage{amsthm}
\usepackage{amssymb}
\usepackage[pdftex,ocgcolorlinks]{hyperref}
\usepackage[shortlabels]{enumitem}
\usepackage{accents}
\usepackage{authblk}
\usepackage[pdftex]{graphicx}

\theoremstyle{plain}
\newtheorem{theorem}{Theorem}[section]
\newtheorem{lemma}[theorem]{Lemma}

\newtheorem{corollary}[theorem]{Corollary}
\newtheorem{definition}[theorem]{Definition}
\newtheorem{remark}[theorem]{Remark}
\newtheorem{example}[theorem]{Example}

\DeclareMathOperator{\subrank}{Q}
\DeclareMathOperator{\asymptoticsubrank}{\undertilde{Q}}
\newcommand{\norm}[1]{\left\|#1\right\|}
\newcommand{\ball}[2]{B_{#1}(#2)}
\newcommand{\ket}[1]{\left|#1\right\rangle}

\newcommand{\ketbra}[2]{\left|#1\middle\rangle\!\middle\langle#2\right|}

\newcommand{\setbuild}[2]{\left\{#1\middle|#2\right\}}
\newcommand{\loccto}[1][]{\xrightarrow{\textnormal{LOCC}}_{#1}}
\DeclareMathOperator{\boundeds}{\mathcal{B}}
\DeclareMathOperator{\states}{\mathcal{S}}
\DeclareMathOperator{\substates}{\mathcal{S}_{\le}}
\DeclareMathOperator{\Tr}{Tr}
\DeclareMathOperator{\support}{supp}
\DeclareMathOperator{\probability}{Pr}
\DeclareMathOperator{\mean}{\mathbb{E}}
\DeclareMathOperator{\GHZ}{GHZ}
\newcommand{\entropy}{H}
\newcommand{\relativeentropy}[3][]{D_{#1}\left(#2\middle\|#3\right)}
\newcommand{\mutualinformation}{I}
\newcommand{\maxentropy}[1][]{H_{\textnormal{max}}^{#1}}
\newcommand{\minentropy}[1][]{H_{\textnormal{min}}^{#1}}
\newcommand{\distributions}[1][]{\mathcal{P}_{#1}}
\newcommand{\subdistributions}{\mathcal{P}_{\le}}
\newcommand{\typeclass}[2]{T^{#1}_{#2}}

\newcommand{\fidelity}{F}
\newcommand{\purifieddistance}{D}

\newcommand{\upentropy}[1][]{H^{\uparrow}_{#1}}
\newcommand{\downentropy}[1][]{H^{\downarrow}_{#1}}
\newcommand{\distillablerate}{E_{D,\textnormal{GHZ}}}
\newcommand{\entanglementcost}{E_{C}}
\newcommand{\entanglementofformation}{E_{F}}

\newcommand{\altminentropy}[1][]{\widehat{H}_{\textnormal{min}}^{#1}}

\newtheorem{problem}{Problem}

\title{Distillation of Greenberger--Horne--Zeilinger states by combinatorial methods}
\author[1,2,3]{P\'eter Vrana}
\author[2]{Matthias Christandl}
\affil[1]{Department of Geometry, Budapest University of Technology and Economics, Egry J\'ozsef u. 1., 1111 Budapest, Hungary}
\affil[2]{QMATH, Department of Mathematical Sciences, University of Copenhagen, Universitetsparken 5, 2100 Copenhagen, Denmark}
\affil[3]{MTA-BME Lend\"ulet Quantum Information Theory Research Group}
\date{\today}

\begin{document}
\maketitle

\begin{abstract}
We prove a lower bound on the rate of Greenberger--Horne--Zeilinger states distillable from pure multipartite states by local operations and classical communication (LOCC). Our proof is based on a modification of a combinatorial argument used in the fast matrix multiplication algorithm of Coppersmith and Winograd. Previous use of methods from algebraic complexity in quantum information theory concerned transformations with stochastic local operations and classical communication (SLOCC), resulting in an asymptotically vanishing success probability. In contrast, our new protocol works with asymptotically vanishing error.
\end{abstract}

\section{Introduction}

When two or more parties are only allowed to operate locally and use classical communication channels (LOCC), entanglement shared between them becomes a resource. This resource plays a central role in quantum information theory, therefore much effort has been put into understanding the possible transformations between different kinds of entangled states under LOCC operations.

Entanglement in bipartite pure states is well understood, and the condition for convertibility becomes particularly simple in the limit of many copies \cite{bennett1996concentrating}. The reason is that any LOCC transformation between such states is asymptotically reversible, therefore there is an essentially unique quantity measuring the amount of entanglement, the von~Neumann entropy of the two reduced states. It is common to choose the base of logarithm to be two, which amounts to choosing the ebit to be the unit of entanglement.

The situation is much more complicated for more than two parties because of the absence of a unique standard state into which any other (pure) state can be reversibly transformed. Computing the conversion rates between any pair of states is certainly out of reach. Instead of this, one can e.g. focus on a specific target state of interest and try to find how many copies of it can be distilled from many copies of an arbitrary state. When the target state is an ebit shared between a specified pair of states, this optimal rate is known as the asymptotic entanglement of assistance and, for pure initial states, coincides with the minimum entanglement entropy over the possible bipartitions separating the members of the specified pair \cite{horodecki2005partial,smolin2005entanglement}. More generally, one can consider products of ebits in some fixed configuration as in the entanglement combing protocol \cite{yang2009entanglement}.

Much less is known when the target state contains genuine multipartite entanglement. The simplest such state is arguably the (multipartite) Greenberger--Horne--Zeilinger (GHZ) state
\begin{equation}
\frac{1}{\sqrt{2}}\left(\ket{00\ldots 0}+\ket{11\ldots 1}\right),
\end{equation}
the key ingredient of the quantum secret sharing protocol \cite{hillery1999quantum}. The method of ref. \cite{smolin2005entanglement} provides a lower bound of simultaneously distilling EPR pairs and GHZ states. Very recently, a combination of the entanglement combing and state merging protocols has been applied to transformations between multipartite entangled states \cite{streltsov2017rates}.

Our result is a new lower bound on the distillable GHZ rate ($\distillablerate$, for a precise definition see Section~\ref{sec:notations}) for pure multipartite states, and can be stated in terms of the joint probability distribution induced by measuring the state in a product orthonormal basis. We need the concept of Shannon conditional entropy
\begin{equation}
\entropy(X|Y)_P=-\sum_{y\in\mathcal{Y}}P_Y(y)\sum_{x\in\mathcal{X}}P_{X|Y}(x|y)\log P_{X|Y}(x|y),
\end{equation}
where $P$ is the joint distribution of $X$ and $Y$.
\begin{theorem}\label{thm:asymptotic}
Let $\ket\psi\in\mathbb{C}^{I_1}\otimes\cdots\otimes\mathbb{C}^{I_k}$ be a unit vector and $P(i_1,\ldots,i_k)=|\psi_{i_1\ldots i_k}|^2$ the associated probability distribution, considered to be the joint distribution of random variables $A_1,\ldots,A_k$. Let $x_1,\ldots,x_k\ge 0$ be such that
\begin{equation*}
\forall J\subseteq[k],J\neq\emptyset,J\neq[k]:\sum_{j\in J}x_j\ge\entropy(A_J|A_{\overline{J}})_P.
\end{equation*}
Then $\displaystyle\distillablerate(\ketbra{\psi}{\psi})\ge\entropy(P)-\sum_{j=1}^kx_j$.
\end{theorem}

The proof is inspired by a recently found connection between entanglement transformations and algebraic complexity theory \cite{chitambar2008tripartite,chen2010tensor,yu2014obtaining,vrana2015asymptotic}. It has been observed that complexity upper bounds on tensor powers of bilinear maps can be directly interpreted as achievability by asymptotic SLOCC transformations. In the one-shot regime, SLOCC transformations were introduced in \cite{bennett2000exact} as a relaxation of LOCC convertibility, while a characterisation in terms of tensor products of linear maps was given in \cite{dur2000three}. In the asymptotic regime LOCC transformations are allowed to introduce a small error, but the notion of asymptotic SLOCC convertibility requires the final state to be reached exactly for any finite number of copies, albeit only with some nonzero probability. For this reason, the two types of conversion rates are incomparable in the sense that any of them can be (arbitrarily) larger than the other, depending on the initial and target states. As a simple example, consider the family of two-qubit pure states $\ket{\psi_p}=\sqrt{p}\ket{00}+\sqrt{1-p}\ket{11}$ ($p\in(0,1)$): the optimal SLOCC rate of the transformation $\ket{\psi_p}$ to $\ket{\psi_q}$ is always $1$, whereas the optimal LOCC rate is $h(p)/h(q)$, which can take any positive value. Thus it is not possible to directly translate results in algebraic complexity theory into bounds on asymptotic LOCC transformations. Nevertheless, we will exhibit a nontrivial SLOCC protocol which can be upgraded to an (asymptotically perfect) LOCC one.

The starting point of our investigation is a combinatorial result from \cite{coppersmith1990matrix}, which forms the basis of the upper bound $2.40363\ldots$ on the exponent of matrix multiplication. The result is a lower bound on the asymptotic subrank of a specific set, which was later generalized to a large family of sets (called tight sets) in \cite{strassen1991degeneration}, where a matching upper bound was also derived. One of the main ideas of the lower bound proof is intersecting (a large power of) the subset with a product of random subsets with a carefully chosen joint distribution, one which makes use of the tightness of the set. We investigate the effect of choosing the distribution in a different, simpler way, which leads to a weaker bound, but one which applies to subsets without such special structure. When applied to the support of a state in a product basis, the asymptotic subrank serves as a lower bound on the rate at which GHZ states can be extracted by asymptotic SLOCC transformations.

We then show how to adapt the idea to LOCC transformations. Instead of taking the intersection of the support with subsets, which would amount to projecting out a large portion of the state, we replace this first step with a randomly chosen measurement and apply the rest of the protocol to the post-measurement state. At the same time, we control the coefficients of the resulting GHZ-like states to estimate the equivalent number of standard GHZ states.

The structure of the paper is as follows. In Section~\ref{sec:subrank} we give a high-level explanation the combinatorial result from \cite{coppersmith1990matrix,strassen1991degeneration} which forms the basis of the upper bound $2.40363\ldots$ on the exponent of matrix multiplication. We present the argument in a form which differs from the original formulation, mainly to separate those ideas that we use in later sections to prove our main result from those that we do not. In Section~\ref{sec:main} we prove our main result. Some properties of our lower bound as well as number of examples are presented in Section~\ref{sec:linearprogram}. In Section~\ref{sec:conclusion} we compare the lower bound with earlier bounds from the literature. Sections \ref{sec:entropies} and \ref{sec:GHZmixtures} contain the proofs of technical lemmas related to the conversion of nonuniform GHZ-like states to approximately uniform ones.

\section{The asymptotic subrank of a subset}\label{sec:subrank}

In this section we review and generalize a technique used by Coppersmith and Winograd \cite{coppersmith1990matrix}, and by Strassen \cite{strassen1991degeneration}, which gives a lower bound on the asymptotic subrank of a subset inside a Cartesian product. This puts Lemma~\ref{lem:meanbounds} into context, which is the only ingredient from this section to be used later.

The asymptotic subrank of such a subset can be viewed as a combinatorial analog of the asymptotic subrank of tensors (not discussed here) as well as the distillable GHZ rate of pure states.
\begin{definition}\label{def:subrank}
Let $I_1,\ldots,I_k$ be finite sets. A subset $\Gamma\subseteq I_1\times\cdots\times I_k$ is called a diagonal if the restriction maps $\pi_j|_{\Gamma}$ are injective for each $j\in\{1,\ldots,k\}$, where $\pi_j:I_1\times\cdots\times I_k\to I_j$ denotes the $j$th projection.

Let $\Phi\subseteq I_1\times\cdots\times I_k$ arbitrary. A subset $\Gamma\subseteq\Phi$ is called free (for $\Phi$) if $\Gamma=\Phi\cap(\pi_1(\Gamma)\times\cdots\times\pi_k(\Gamma))$.

The subrank $\subrank(\Phi)$ is the size of the largest free diagonal $\Gamma\subseteq\Phi$. For $\Phi\subseteq I_1\times\cdots\times I_k$ and $\Psi\subseteq J_1\times\cdots\times J_k$ we define the product $\Phi\times\Psi\subseteq(I_1\times J_1)\times\cdots\times(I_k\times J_k)$. The asymptotic subrank is $\asymptoticsubrank(\Phi)=\lim_{n\to\infty}\subrank(\Phi^{\times n})^{1/n}=\sup\subrank(\Phi^{\times n})^{1/n}$.
\end{definition}
For example, the support of an $r$-level tripartite generalized GHZ state in the usual basis is $\{(1,1,1),(2,2,2),\ldots,(r,r,r)\}$, which is itself a diagonal of size $r$, therefore has subrank $r$. On the other hand, the support of the W state $\frac{1}{\sqrt{3}}(\ket{100}+\ket{010}+\ket{001})$ is $\{(1,0,0),(0,1,0),(0,0,1)\}$, which has subrank $1$.

The key steps in the Coppersmith--Winograd--Strassen lower bound method can be described as follows:
\begin{enumerate}[1.]
\item Draw subsets $W_j\subseteq I_j$ ($j=1,\ldots,k$) from some distribution (possibly in a correlated way).
\item Consider the (random) graph $G=(V,E)$ with $V=\Phi\cap(W_1\times\cdots\times W_k)$ and
\begin{equation*}
E=\setbuild{\{(i_1,\ldots,i_k),(i'_1,\ldots,i'_k)\}\in\binom{V}{2}}{\exists j:i_j=i'_j}.
\end{equation*}
\item The set $\Gamma$ of isolated vertices in $G$ is a free diagonal.
\item Bound $\mean|\Gamma|$ from below and use $\mean|\Gamma|\le\subrank(\Phi)$.
\end{enumerate}
In refs. \cite{coppersmith1990matrix,strassen1991degeneration} the set $\Phi$ is a truncation of a large power of the tight (see \cite[Section 5.]{strassen1991degeneration} for the definition) set $\Psi$ and the joint distribution of the subsets $W_j$ is carefully chosen accordingly. We do not wish to make such restrictions at this point, but we will assume the following property which simplifies the calculations considerably.
\begin{definition}\label{def:homogeneous}
We say that the joint distribution of $W_1,\ldots,W_k$ is homogeneous for $\Phi$ if $\probability[i_1,i'_1\in W_1,\ldots,i_k,i'_k\in W_k]$ depends on $i_1,i'_1,\ldots,i_k,i'_k$ only through the subset $J:=\setbuild{j\in[k]}{i_j\neq i'_j}$ when $(i_1,\ldots,i_k)\in\Phi$ and $(i'_1,\ldots,i'_k)\in\Phi$.
\end{definition}

The following lemma is the core of the argument. For bounding the subrank, only the lower bound is needed and even that only in the special case when $f$ is the constant $1$ function. However, we will need the general form later, and the proof is essentially the same.
\begin{lemma}\label{lem:meanbounds}
Let $W_j\subseteq I_j$ be random subsets with distribution homogeneous for some $\Phi\subseteq I_1\times\cdots\times I_k$. Let $p_J$ denote the common value of $\probability[i_1,i'_1\in W_1,\ldots,i_k,i'_k\in W_k]$ when $J:=\setbuild{j\in[k]}{i_j\neq i'_j}$. Consider the set $\Gamma$ of isolated vertices in the random graph $G=(V,E)$ as introduced before. Then for any function $f:\Phi\to\mathbb{R}_+$ the following estimates hold:
\begin{equation}
\left(p_{\emptyset}-\sum_{\substack{J\subseteq[k]  \\  J\neq\emptyset,J\neq[k]}}2^{\maxentropy(A_J|A_{\overline{J}})_\Phi}p_J\right)\sum_{i\in\Phi}f(i)\le\mean\sum_{i\in\Gamma}f(i)\le p_{\emptyset}\sum_{i\in\Phi}f(i).
\end{equation}
Here the max-entropy refers to any random variable with support equal to $\Phi$, i.e.
\begin{equation}
\maxentropy(A_J|A_{\overline{J}})_\Phi=\log\max_{(i_j)_{j\in \overline{J}}}\left|\setbuild{(i_j)_{j\in J}}{(i_1,\ldots,i_k)\in\Phi}\right|.
\end{equation}
\end{lemma}
\begin{proof}
For any realization of the random graph $G=(V,E)$ the following inequalitites hold (note that the second sum is over unordered pairs):
\begin{equation}
\sum_{i\in V}f(i)-\sum_{\{i,i'\}\in E}\left(f(i)+f(i')\right)\le\sum_{i\in\Gamma}f(i)\le\sum_{i\in V}f(i).
\end{equation}
This implies that similar relations are true for the expected values, i.e.
\begin{equation}
\mean\sum_{i\in V}f(i)-\mean\sum_{\{i,i'\}\in E}\left(f(i)+f(i')\right)\le\mean\sum_{i\in\Gamma}f(i)\le\mean\sum_{i\in V}f(i).
\end{equation}
The first sum can be computed as
\begin{equation}
\mean\sum_{i\in V}f(i) = \sum_{i\in \Phi}\probability[i_1\in W_1,\ldots,i_k\in W_k]f(i) = \sum_{i\in \Phi}p_{\emptyset}f(i),
\end{equation}
while the second one can be bounded as
\begin{equation}
\begin{split}
\mean\sum_{\{i,i'\}\in E}\left(f(i)+f(i')\right)
 & = \sum_{\substack{\{i,i'\}\in\binom{\Phi}{2}  \\  \exists j:i_j=i'_j}}\probability[i_1,i'_1\in W_1,\ldots,i_k,i'_k\in W_k]\left(f(i)+f(i')\right)  \\
 & = \sum_{i\in\Phi}\sum_{\substack{i'\in\Phi\setminus\{i\}  \\  \exists j:i_j=i'_j}}\probability[i_1,i'_1\in W_1,\ldots,i_k,i'_k\in W_k]f(i)  \\
 & = \sum_{i\in\Phi}\sum_{\substack{J\subseteq[k]  \\  J\neq\emptyset,J\neq[k]}}\sum_{\substack{i'\in\Phi  \\  \setbuild{j}{i_j\neq i'_j}=J}}\probability[i_1,i'_1\in W_1,\ldots,i_k,i'_k\in W_k]f(i)  \\
 & = \sum_{\substack{J\subseteq[k]  \\  J\neq\emptyset,J\neq[k]}}\sum_{i\in\Phi}\sum_{\substack{i'\in\Phi  \\  \setbuild{j}{i_j\neq i'_j}=J}}p_Jf(i)  \\
 & \le \sum_{\substack{J\subseteq[k]  \\  J\neq\emptyset,J\neq[k]}}\sum_{i\in\Phi}2^{\maxentropy(A_J|A_{\overline{J}})_\Phi}p_Jf(i).
\end{split}
\end{equation}
The second equality uses that the sum over unordered pairs is half the sum over distinct ordered pairs, and that the sum of $f(i)$ is equal to the sum of $f(i')$. In the next step the sum over $i'$ is split according to the location of the components shared with $i$.
\end{proof}

By construction, $\Gamma$ is a free diagonal in $\Phi$. The expected size of $\Gamma$ is a lower bound on the maximum size of the free diagonal. Applying the lower bound of Lemma~\ref{lem:meanbounds} for the function $f(i)=1$ leads to the estimate
\begin{equation}\label{eq:subranklowerbound}
\subrank(\Phi)\ge|\Phi|\left(p_{\emptyset}-\sum_{\substack{J\subseteq[k]  \\  J\neq\emptyset,J\neq[k]}}2^{\maxentropy(A_J|A_{\overline{J}})_\Phi}p_J\right).
\end{equation}

For lower bounding the asymptotic subrank, one chooses $\Phi=\Psi^{\times n}\cap(\typeclass{n}{P_1}\times\cdots\times\typeclass{n}{P_k})$ for some type classes where $P_j$ has a limit. In such a setting the conditional max-entropy grows linearly. Note also that $\log|\Phi|=\maxentropy(A_{[k]}|A_{\emptyset})_{\Phi}$.
\begin{lemma}\label{lem:maxentropyrate}
Let $\Psi\subseteq I_1\times\cdots\times I_k$ and $P\in\distributions(\Psi)$ (the set of probability distributions on $\Psi$). Let $\Phi^{(n)}=\Psi^{\times n}\cap(\typeclass{n}{P^{(n)}_1}\times\cdots\times\typeclass{n}{P^{(n)}_k})$ for some types $P^{(n)}\in\distributions[n](\Psi)$ such that $P^{(n)}\to P$. Then
\begin{equation}
h_J:=\lim_{n\to\infty}\frac{1}{n}\maxentropy(A_J|A_{\overline{J}})_{\Phi^{(n)}}=\max_{\substack{Q\in\distributions(\Psi)  \\  \forall j:Q_j=P_j}}\entropy(A_J|A_{\overline{J}})_Q.
\end{equation}
\end{lemma}
\begin{proof}
$\Phi^{(n)}$ is a disjoint union of type classes
\begin{equation}
\Phi^{(n)}=\bigcup_Q\typeclass{n}{Q}
\end{equation}
where $Q\in\distributions[n](\Psi)$ such that $Q_j=P^{(n)}_j$ for all $j$. The exponentiated max-entropy is monotone in the support and subadditive under taking unions, therefore
\begin{equation}
\max_Q\maxentropy(A_J|A_{\overline J})_{\typeclass{n}{Q}}\le\maxentropy(A_J|A_{\overline J})_{\Phi^{(n)}}\le\max_Q\maxentropy(A_J|A_{\overline J})_{\typeclass{n}{Q}}+\log|\distributions[n](\Psi)|
\end{equation}
Consider the projection $(I_1\times\cdots\times I_k)^{\times n}\to(\prod_{j\in\overline J}I_j)^{\times n}$. This map is equivariant under the $S_n$ action permuting the factors. Its restriction to $\typeclass{n}{Q}$ is onto the set $\typeclass{n}{Q_{\overline J}}$, where $S_n$ acts transitively. Therefore
\begin{equation}
\maxentropy(A_J|A_{\overline J})_{\typeclass{n}{Q}}=\log\frac{|\typeclass{n}{Q}|}{|\typeclass{n}{Q_{\overline J}}|},
\end{equation}
so
\begin{equation}
n\entropy(A_J|A_{\overline{J}})_Q-\log|\distributions[n](\Psi)|\le\maxentropy(A_J|A_{\overline J})_{\typeclass{n}{Q}}\le n\entropy(A_J|A_{\overline{J}})_Q+\log|\distributions[n](\Psi)|.
\end{equation}
The claim follows by continuity and using $\log|\distributions[n](\Psi)|=o(n)$.
\end{proof}

The following example was the main motivation for our work, but it is not necessary for understanding our results. The reader may wish to skip to Example~\ref{ex:uniformrandom}, which is simpler and more similar to our main theorem.
\begin{example}[Asymptotic subrank of tight sets, $k=3$]\label{ex:tight}
This example is by Strassen from ref. \cite{strassen1991degeneration}, based on the ideas of ref. \cite{coppersmith1990matrix}. We present it in a different but equivalent form and omit some details.

Let $\Psi\subseteq I_1\times I_2\times I_3$ be tight, $P\in\distributions[n](\Psi)$ and take a sequence $P^{(n)}\in\distributions[n](\Psi)$ let $\Phi^{(n)}=\Psi^{\times n}\cap(\typeclass{n}{P^{(n)}_1}\times\typeclass{n}{P^{(n)}_2}\times\typeclass{n}{P^{(n)}_3})$. Without loss of generality assume that $P$ has maximal entropy given its marginal distributions. Then
\begin{subequations}\label{eq:tightgrowthrates}
\begin{align}
|\Phi^{(n)}| & =2^{n\entropy(P)+o(n)}  \\
\maxentropy(A_j|A_{\overline{j}})_{\Phi^{(n)}} & = 0  \\
\begin{split}
\maxentropy(A_{\overline{j}}|A_j)_{\Phi^{(n)}} & = n\left(\entropy(P)-\entropy(P_j)\right)+o(n)
\end{split}
\end{align}
\end{subequations}
by Lemma~\ref{lem:maxentropyrate} and tightness.

For a large prime $M\in\mathbb{N}$, draw $a_j:I_j^{\times n}\to\mathbb{Z}_M$ uniformly from the space of triples of functions satisfying $a_1(i_1)+a_2(i_2)=2a_3(i_3)$ for all $(i_1,i_2,i_3)\in\Phi^{(n)}$. Let $S\subseteq\mathbb{Z}_M$ be a subset of $\{0,1,\ldots,\frac{M-1}{2}\}$ without three-term arithmetic progressions, and let $W_j=a_j^{-1}(S)$. It can be shown that the joint distribution is homogeneous for $\Phi^{(n)}$. More precisely,
\begin{equation}
p_{\emptyset} = \frac{|S|}{M^2}\qquad\text{and}\qquad p_{\{1,2\}}=p_{\{1,3\}}=p_{\{2,3\}} = \frac{|S|}{M^3}
\end{equation}
Choose $M=2^{n(\entropy(P)-\min\{\entropy(P_1),\entropy(P_2),\entropy(P_3)\})+o(n)}$ and let $S$ be as large as possible. Then $|S|=M^{1-o(1)}$ as shown in ref. \cite{salem1942sets}. Inserting these as well as the asymptotics from \eqref{eq:tightgrowthrates} into \eqref{eq:subranklowerbound} gives
\begin{equation}
\begin{split}
\subrank(\Psi^{\times n})
 & \ge \subrank(\Phi^{(n)})  \\
 & \ge |\Phi^{(n)}|\left(\frac{|S|}{M^2}-2^{o(n)}\left(2^{n(\entropy(P)-\entropy(P_1))}+2^{n(\entropy(P)-\entropy(P_2))}+2^{n(\entropy(P)-\entropy(P_3))}\right)\frac{|S|}{M^3}\right)  \\
 & \ge |\Phi^{(n)}|\frac{|S|}{M^2}\left(1-2^{n(\entropy(P)-\min\{\entropy(P_1),\entropy(P_2),\entropy(P_3)\})+o(n)}\frac{1}{M}\right)  \\
 & = 2^{n\entropy(P)-(n(\entropy(P)-\min\{\entropy(P_1),\entropy(P_2),\entropy(P_3)\})+o(n))}  \\
 & =2^{n\min\{\entropy(P_1),\entropy(P_2),\entropy(P_3)\}+o(n)},
\end{split}
\end{equation}
which implies $\log\asymptoticsubrank(\Psi)\ge\min\{\entropy(P_1),\entropy(P_2),\entropy(P_3)\}$.
\end{example}

\begin{example}\label{ex:uniformrandom}
Let $\Phi\subseteq I_1\times\cdots\times I_k$ be arbitrary and choose $W_1,\ldots,W_k$ by including in $W_j$ each element of $I_j$ with probability $q_j$, independently of all other choices. Then
\begin{equation}
p_J=q_1q_2\cdots q_k\prod_{j\in J}q_j,
\end{equation}
where the first $k$ factors are the probabilities of the (independent) events that $i_j\in W_j$, while the remaining ones correspond to $i'_j\in W_j$ for $j\in J$. Therefore
\begin{equation}
\subrank(\Phi)\ge|\Phi|q_1q_2\cdots q_k\left(1-\sum_{\substack{J\subseteq[k]  \\  J\neq\emptyset,J\neq[k]}}2^{\maxentropy(A_J|A_{\overline{J}})_\Phi+\sum_{j\in J}\log q_j}\right).
\end{equation}

Now let $\Psi\subseteq I_1\times\cdots\times I_k$ be arbitrary, $P\in\distributions(\Psi)$ and $\Phi^{(n)}=\Psi^{\times n}\cap(\typeclass{n}{P^{(n)}_1}\times\cdots\times\typeclass{n}{P^{(n)}_k})$ for some distributions $P^{(n)}\in\distributions[n](\Psi)$ converging to $P$. Let $h_J$ be as in Lemma~\ref{lem:maxentropyrate} and choose $q_j$ as $q_j=2^{-nx_j}$ for some real numbers $x_j>0$. Then
\begin{equation}\label{eq:subrankfromuniformrandom}
\subrank(\Psi^{\times n})\ge\subrank(\Phi^{(n)})\ge 2^{n(h_{[k]}-\sum_{j=1}^k x_j)+o(n)}\left(1-\sum_{\substack{J\subseteq[k]  \\  J\neq\emptyset,J\neq[k]}}2^{n\left(h_J-\sum_{j\in J}x_j\right)+o(n)}\right).
\end{equation}
As long as $h_J<\sum_{j\in J}x_j$ for each subset $J$ (except for $\emptyset$ or $[k]$), the sum in the second factor of \eqref{eq:subrankfromuniformrandom} goes to $0$ as $n\to\infty$. If this holds, then we get the estimate
\begin{equation}
\log\asymptoticsubrank(\Psi)\ge\frac{1}{n}\log\subrank(\Psi^{\times n})\ge h_{[k]}-\sum_{j=1}^k x_j.
\end{equation}
\end{example}

\section{Distillation of asymptotically perfect GHZ states with LOCC}\label{sec:main}

\subsection{Proof strategy}

In this section we prove our main result (Theorem~\ref{thm:asymptotic}), a lower bound on the multiparty distillable entanglement ($\distillablerate$) of an arbitrary pure state, which can be seen as an analog of Example~\ref{ex:uniformrandom}.

Let $\ket{\varphi}\in\mathcal{H}_1\otimes\cdots\otimes\mathcal{H}_k$ with finite dimensional Hilbert spaces $\mathcal{H}_j$, assume that $\norm{\ket{\varphi}}=1$. Choose an orthonormal basis in each of the Hilbert spaces (henceforth identified as $\mathcal{H}_j\simeq\mathbb{C}^{I_j}$), and let $\Phi=\support\ket{\varphi}\subseteq I_1\times\cdots\times I_k$ be the support. The methods of Section~\ref{sec:subrank} can be applied to extract a generalized GHZ state if a local restriction to a subset $W$ is interpreted as performing a two-outcome measurement with projections $\Pi_W,I-\Pi_W$, where
\begin{equation}
\Pi_W=\sum_{i\in W}\ketbra{i}{i}.
\end{equation}
This leads to a generalized GHZ state of rank $Q(\Phi)$, but the success probability may be very low and there is no control on the coefficients of the resulting state, i.e. we only get an asymptotic SLOCC transformation. Our goal in this section is to improve the protocol such that the success probability is close to $1$.

To understand the reason for the loss of probability, first note that the randomized construction from Section~\ref{sec:subrank} actually gives rise to a two-step protocol. In the first step the parties project onto the subspaces generated by $W_1,\ldots,W_k$, while in the second step they project again onto $\pi_1(\Gamma),\ldots,\pi_k(\Gamma)$. Suppose for a moment that the magnitude of the nonzero coefficients of $\ket{\varphi}$ are the same (this is approximately true in a precise sense for e.g. large tensor powers), so that the success probability is proportional to the number of coefficients not projected out. Examining the bounds of Lemma~\ref{lem:meanbounds}, one can see that when the $p_J$ ($J\neq\emptyset$) are sufficiently small, then the size of $\Gamma$ is essentially the same as that of $\support\ket{\varphi}\cap(W_1\times\cdots\times W_k)$. Thus the probability of failure in the second step is negligible.

In order to achieve high probability in the first step, we replace the projection onto a subspace with a measurement with respect to a partition of the set of basis states. Let us choose set partitions $(W_{j,m_j})_{m_j=1}^{M_j}$ of $I_j$. For each $m=(m_1,\ldots,m_k)\in[M_1]\times\cdots\times[M_k]$ consider the graphs $G_m=(V_m,E_m)$ with $V_m=\Phi\cap(W_{1,m_1}\times\cdots\times W_{k,m_k})$ and
\begin{equation}
E_m=E\cap\binom{V_m}{2},
\end{equation}
i.e. two elements in the support are adjacent if they share at least one coordinate. Let $\Gamma_m$ be the set of isolated vertices in $G_m$. We use the following improved protocol.
\begin{enumerate}
\item At site $j$ perform a measurement according to the pairwise orthogonal projections
\begin{equation}
\Pi_{j,m_j}=\sum_{i\in W_{j,m_j}}\ketbra{i}{i}.
\end{equation}
\item If the outcomes are $m=(m_1,\ldots,m_k)$, then the resulting state has support $V_m$, which has a free diagonal $\Gamma_m$. Extract this diagonal by measuring the local projections
\begin{equation}
\sum_{i\in \pi_j(\Gamma_m)}\ketbra{i}{i}.
\end{equation}
\item If every measurement is succesful, then the resulting state is a generalized GHZ state.
\end{enumerate}
If the parties wish to distill (standard) GHZ states from many copies, then after running the above protocol many times, the obtained states can be converted into GHZ states at a rate given by the von~Neumann entropy of its reduced states, averaged over the measurement outcomes. One may think of this expected entropy as the asymptotic value of the ensemble.

Our goal is to show that a randomly chosen $k$-tuple of partitions leads to a good estimate on the asymptotic value of the diagonals. We will work explicitly with mixtures of pairwise orthogonal generalized GHZ states (a parameterization in terms of joint distributions is introduced below) to represent the extracted diagonal. The asymptotic value of such a state is given by the Shannon conditional entropy.

However, as we have already seen in Examples~\ref{ex:tight} and~\ref{ex:uniformrandom}, it can be advantageous to apply the protocol to a state (subset) which is not exactly a power in order to get the best asymptotic bound. In the case of subrank, the high power was first intersected with a product of type classes, which is clearly not affordable in the present case as it would already lead to an asymptotically vanishing success probability. One could intersect instead with typical subsets corresponding to the marginals, but it is possible to do better, since we allow approximate transformations. The conditional max-entropies appearing in the lower bound in Lemma~\ref{lem:meanbounds} can be lowered to approximately the Shannon conditional entropy if we work with a nearby state instead.

One proof strategy would be to derive a lower bound on the Shannon conditional entropy (fairly simple using Lemma~\ref{lem:meanbounds}) and then use the asymptotic equipartition property (AEP) to get an asymptotic statement. The difficulty is that if we apply the protocol to a state deviating slightly from $\ketbra{\psi}{\psi}^{\otimes n}$ and consider many copies of the resulting state, then a qualitative AEP is not useful since the small error gets amplified as we take more and more copies. It is possible to remedy the situation with a quantitative form of the AEP (e.g. the one from \cite{holenstein2011randomness}), but this approach leads to a fairly complicated proof and no useful single-shot bound.

Instead, we derive a lower bound on the (smooth) conditional min-entropy, which governs the number of GHZ states that can be (approximately) extracted in a one-shot setting (see Lemma~\ref{lem:exactfrommixture} for a precise statement). Since its definition involves an optimization over the measurement outcomes, estimating it directly in such a probabilistic setting seems difficult. To circumvent this problem, we use the conditional R\'enyi entropy as an intermediate quantity, which is more convenient to use. In particular, the expected value can again be estimated using Lemma~\ref{lem:meanbounds}.

It is conceivable that with some probability (over the choice of the random set partitions) the conditional min-entropy itself is large, which would lead to our result more directly. However, proving this (if true) would likely involve a measure concentration argument which relies on some strong independence property of the randomly chosen partitions. In contrast, our proof assumes only a homogeneity property (Definition~\ref{def:homogeneous}), which is a condition involving at most $2k$th moments (of the indicator functions of subset memberships), therefore works for a large class of joint distributions.

\subsection{Notations}\label{sec:notations}

We let $\substates(\mathcal{H})=\setbuild{\rho\in\boundeds(\mathcal{H})}{\rho\ge 0,\Tr\rho\le 1}$ be the set of subnormalized states on the Hilbert space $\mathcal{H}$, while the set of normalized states is $\states(\mathcal{H})=\setbuild{\rho\in\boundeds(\mathcal{H})}{\rho\ge 0,\Tr\rho=1}$. For $\rho,\sigma\in\substates(\mathcal{H})$, the purified distance is defined as $\purifieddistance(\rho,\sigma)=\sqrt{1-\fidelity(\rho,\sigma)^2}$, where
\begin{equation}\label{eq:fidelity}
\fidelity(\rho,\sigma)=\sqrt{(1-\Tr\rho)(1-\Tr\sigma)}+\Tr\sqrt{\sigma^{1/2}\rho\sigma^{1/2}},
\end{equation}
is the generalized fidelity \cite[Definitions 2. and 4.]{tomamichel2010duality}. The closed $\epsilon$-ball around a state $\rho$ is $\ball{\epsilon}{\rho}=\setbuild{\rho'\in\substates(\mathcal{H})}{\purifieddistance(\rho,\rho')\le\epsilon}$. We will also write $\rho\approx_\epsilon\sigma$ if $\purifieddistance(\rho,\sigma)\le\epsilon$. 

Probability distributions on a finite set $\mathcal{X}$ will be identified with those states on $\mathbb{C}^{\mathcal{X}}$ which are diagonal in the standard basis. The diagonal elements in $\substates(\mathbb{C}^{\mathcal{X}})$ will be denoted by $\subdistributions(\mathcal{X})$. When considering probability distributions, the $\epsilon$-ball is understood to be $\ball{\epsilon}{P}=\setbuild{P'\in\subdistributions(\mathcal{X})}{\purifieddistance(P,P')\le\epsilon}$.

For multipartite states $\rho\in\substates(\mathcal{H}_1\otimes\cdots\otimes\mathcal{H}_k)$ and $\sigma\in\substates(\mathcal{K}_1\otimes\cdots\otimes\mathcal{K}_k)$ we write $\rho\loccto[\epsilon]\sigma$ if there is a trace-nonincreasing LOCC channel $\Lambda$ such that $\Lambda(\rho)\in\ball{\epsilon}{\sigma}$, while $\rho\loccto\sigma$ means that there is a $\Lambda$ such that $\Lambda(\rho)=\sigma$. Note in particular that $\rho\approx_\epsilon\sigma$ implies $\rho\loccto[\epsilon]\sigma$, since the identity is an LOCC channel. Trace-nonincreasing channels are contractions with respect to the purified distance, therefore the relations $\loccto[\epsilon]$ enjoy a transitivity-like property:
\begin{equation}
(\rho\loccto[\epsilon_1]\sigma\text{ and }\sigma\loccto[\epsilon_2]\tau)\implies(\rho\loccto[\epsilon_1+\epsilon_2]\tau).
\end{equation}
We define the distillable entanglement as
\begin{equation}
\distillablerate(\rho)=\lim_{\epsilon\to 0}\limsup_{n\to\infty}\max\setbuild{\frac{N}{n}}{\rho^{\otimes n}\loccto[\epsilon]\GHZ^{\otimes N}}.
\end{equation}

For joint distributions $P_{XY}\in\distributions(\mathcal{X}\times\mathcal{Y})$, one of the several different notions of the conditional R\'enyi entropy is defined as \cite{arimoto1977information} (see also \cite[Definition 5.2.]{tomamichel2015quantum})
\begin{equation}
\upentropy[\alpha](X|Y)_P=\sup_{Q\in\distributions(\mathcal{Y})}\frac{1}{1-\alpha}\log\left(\sum_{\substack{x\in\mathcal{X}  \\  y\in\mathcal{Y}}}P_{XY}(x,y)^\alpha Q(y)^{1-\alpha}\right).
\end{equation}

\subsection{Proof of the main result}

The protocol explained above leads to a random GHZ-like state. We find it convenient to work with such an output ensemble as a mixed state where the different outcomes are distinguished by classical flags, available to every party. First we define a parameterization of such mixtures by joint distributions.
\begin{definition}\label{def:randomGHZ}
For a nonnegative function $P_{XY}:\mathcal{X}\times\mathcal{Y}\to\mathbb{R}_+$ where $\mathcal{X}$ and $\mathcal{Y}$ are finite sets, we define the unnormalized state
\begin{equation}
\GHZ_{P_{XY}}=\sum_{\substack{y\in\mathcal{Y}  \\  x,x'\in\mathcal{X}}}\sqrt{P_{XY}(x,y)P_{XY}(x',y)}\ketbra{(xy)\ldots(xy)}{(x'y)\ldots(x'y)}
\end{equation}
on the Hilbert space $\left(\mathbb{C}^{\mathcal{X}\times\mathcal{Y}}\right)^{\otimes k}$. States of this form will be referred to as random GHZ states.

A distribution $P\in\distributions(\mathcal{X})$ can be identified with one on $\mathcal{X}\times\{0\}$, and the corresponding (pure) generalized GHZ state will be denoted by $\GHZ_P$. We will write $\GHZ$ to mean $\GHZ_P$ with $P$ the uniform distribution on $\mathcal{X}=\{0,1\}$.
\end{definition}
The role of $X$ and $Y$ in the definition is not symmetric, but reflects the quantum-classical splitting of the state. The second marginal $P_Y$ (if present) encodes the weights in the classical mixture, while the conditional states are pure generalized GHZ states with coefficients given by $\sqrt{P_{X|Y=y}}$.

Asymptotically, $\GHZ_{P_{XY}}$ is equivalent to $\entropy(X|Y)_P$ copies of the $\GHZ$ state. In the single shot regime, the Shannon entropy is not meaningful, but the R\'enyi entropies with $\alpha>1$ can be used to bound the number of $\GHZ$ states that can be approximately extracted from $\GHZ_{P_{XY}}$. This is made precise in the following lemma (the proof is in Section~\ref{sec:GHZmixtures}).
\begin{lemma}\label{lem:approximatefrommixture}
For any $P_{XY}\in\distributions(\mathcal{X}\times\mathcal{Y})$, $\epsilon\in(0,1)$ and $\alpha>1$ the relation $\GHZ_{P_{XY}}\loccto[\epsilon]\GHZ^{\otimes N}$ holds with
\begin{equation}
N=\left\lfloor\upentropy[\alpha](X|Y)_P-\left(1+\frac{1}{\alpha-1}\right)\log\frac{10}{\epsilon^2}\right\rfloor
\end{equation}
\end{lemma}

\begin{lemma}\label{lem:todiagonals}
Let $\ket{\varphi}$ be a unit vector in $\mathbb{C}^{I_1}\otimes\cdots\otimes\mathbb{C}^{I_k}$, let $(W_{j,m_j})_{m_j=1}^{M_j}$ a partition of $I_j$ for $j=1,\ldots,k$, and let $\Gamma_m=\Gamma_{m_1\ldots m_k}$ be a free diagonal in $(\support\ket{\varphi})\cap(W_{1,m_1}\times\cdots\times W_{k,m_k})$, and let $R\in\distributions(\mathcal{Y})$ be a ``reference'' distribution with $\mathcal{Y}=([M_1]\times\cdots\times [M_k])\cup\{*\}$. Then for any $\alpha>1$ there exists a probability distribution $P_{XY}$ on $\mathcal{X}\times\mathcal{Y}$ (for some finite set $\mathcal{X}$) such that $\ketbra{\varphi}{\varphi}\loccto\GHZ_{P_{XY}}$ and
\begin{equation}\label{eq:achievablevalue}
\upentropy[\alpha](X|Y)_P\ge\frac{1}{1-\alpha}\log\left[\sum_{m}\sum_{i\in\Gamma_m}|\varphi_i|^{2\alpha}R(m)^{1-\alpha}+\left(1-\sum_{m}\sum_{i\in\Gamma_m}|\varphi_i|^{2}\right)^\alpha R(*)^{1-\alpha}\right]
\end{equation}
\end{lemma}
\begin{proof}
We use the protocol explained above. For each $j\in[k]$, the $j$th party performs a measurement with operators
\begin{equation}
\Pi_{j,m_j}=\sum_{i\in W_{j,m_j}}\ketbra{i}{i},
\end{equation}
and broadcasts the outcome $m_j$ to all the other parties. This results in the state
\begin{equation}
\sum_{m}(\Pi_{1,m_1}\otimes\cdots\otimes\Pi_{k,m_k})\ketbra{\varphi}{\varphi}(\Pi_{1,m_1}\otimes\cdots\otimes\Pi_{k,m_k})\otimes\ketbra{mm\ldots m}{mm\ldots m}.
\end{equation}
In the next step, every party performs a two-outcome measurement conditioned on the measurement result $m$. One of the operators at party $j$ is
\begin{equation}
\sum_{i\in\Gamma_m}\ketbra{\pi_j(i)}{\pi_j(i)},
\end{equation}
associated with the outcome ``success'', while the other outcome is interpreted as ``failure''. $\Gamma_m$ is a free diagonal in the support of the conditional state, therefore if every measurement is successful, then after applying the local partial isometries
\begin{equation}
\sum_{i\in\Gamma_m}\ketbra{i}{\pi_j(i)},
\end{equation}
the parties end up with a generalized GHZ state. The phases of the coefficients can now be adjusted with a diagonal unitary applied by any of the parties. From now on we assume that the resulting coefficients are nonnegative real numbers. In this case, they keep the flag $m$ and the protocol is finished. Otherwise, if any of the measurements fails, they discard the measurement result and prepare the separable state
\begin{equation}
\ketbra{00\ldots 0}{00\ldots 0}\otimes\ketbra{**\ldots*}{**\ldots *},
\end{equation}
where the second factor represents the new value of the flag. The protocol clearly implements the transformation $\ketbra{\varphi}{\varphi}\loccto\GHZ_{P_{XY}}$ where $\mathcal{X}=(I_1\times\cdots\times I_k)\cup\{0\}$, and $P\in\distributions(\mathcal{X}\times\mathcal{Y})$ is defined as
\begin{equation}
P_{XY}(x,y)=\begin{cases}
|\varphi_x|^2  &  \text{if $y\in [M_1]\times\cdots\times [M_k]$ and $x\in\Gamma_y$}  \\
1-\sum_{m}\sum_{i\in\Gamma_m}|\varphi_i|^{2} & \text{if $y=*$ and $x=0$}  \\
0 & \text{otherwise.}
\end{cases}
\end{equation}

Finally, $\upentropy[\alpha](X|Y)_P$ is defined as a supremum over distributions on $\mathcal{Y}$, therefore any given $R\in\distributions(\mathcal{Y})$ provides a lower bound as stated.
\end{proof}

Now we can prove the single-shot form of our main theorem.
\begin{theorem}\label{thm:oneshot}
Let $\ket\varphi\in\mathbb{C}^{I_1}\otimes\cdots\otimes\mathbb{C}^{I_k}$ be a unit vector and $Q(i_1,\ldots,i_k)=|\varphi_{i_1\ldots i_k}|^2$ the associated probability distribution, considered to be the joint distribution of random variables $A_1,\ldots,A_k$. Let $M_1,\ldots,M_k\ge 1$ be integers and define
\begin{equation}
\Delta=-k+\min_{\substack{J\subseteq[k]  \\  J\neq\emptyset,J\neq[k]}}\left(\sum_{j\in J}\log M_j-\maxentropy(A_J|A_{\overline{J}})_Q\right).
\end{equation}
Then for any $\epsilon\in(0,1)$ and $\alpha>1$ the relation $\ketbra{\varphi}{\varphi}\loccto[\epsilon]\GHZ^{\otimes N}$ holds with
\begin{equation}\label{eq:oneshotbound}
N=\left\lfloor \frac{\alpha}{1-\alpha}\log\left(2^{\frac{1-\alpha}{\alpha}\left(\entropy_\alpha(Q)-\sum_{j=1}^k\log M_j\right)}+2^{-\Delta/\alpha}\right)-\left(1+\frac{1}{\alpha-1}\right)\log\frac{10}{\epsilon^2}\right\rfloor.
\end{equation}
\end{theorem}
\begin{proof}
For each $j\in[k]$ choose the partition $(W_{j,m_j})_{m_j=1}^{M_j}$ randomly in the following way: for each element in $I_j$, draw a label uniformly at random from $[M_j]$ independently of all other choices, and let $W_{j,m_j}$ be the set of elements having label $m_j$. Let $V_m=(\support\ket{\varphi})\cap(W_{1,m_1}\times\cdots\times W_{k,m_k})$ (where $m=(m_1,\ldots,m_k)$), and consider the set $\Gamma_m$ of isolated vertices in the graph $G_m=(V_m,E_m)$, where the edges are the colliding pairs. We will use Lemma~\ref{lem:todiagonals} with these partitions and diagonals and with the probability distribution $R(m)=(1-r)(M_1M_2\cdots M_k)^{-1}$, $R(*)=r$ for some $r\in(0,1)$ chosen later. In the following we derive a lower bound on the mean of the right hand side of \eqref{eq:achievablevalue}. Note that for any $m$ the joint distribution of $(W_{1,m_1},\cdots,W_{k,m_k})$ is homogeneous for $\support\ket{\varphi}$ with $p_J=(M_1\cdots M_k)^{-1}\prod_{j\in J}M_j^{-1}$.

First note that $x\mapsto-\frac{1}{\alpha-1}\log x$ is convex and decreasing, therefore it is enough to find an upper bound on the expected value of its argument,
\begin{equation}\label{eq:suminlog}
\sum_m\sum_{i\in\Gamma_m}Q(i)^\alpha\left(\frac{1-r}{M_1M_2\cdots M_k}\right)^{1-\alpha}+\left(1-\sum_m\sum_{i\in\Gamma_m}Q(i)\right)^\alpha r^{1-\alpha}.
\end{equation}

For the first term, we use the upper bound from Lemma~\ref{lem:meanbounds} (with $\Phi=\support\ket{\varphi}$ and $f(i)=Q(i)^\alpha$):
\begin{equation}\label{eq:firstterm}
\begin{split}
\mean\sum_m\sum_{i\in\Gamma_m}Q(i)^\alpha\left(\frac{1-r}{M_1M_2\cdots M_k}\right)^{1-\alpha}
 & =\left(\frac{1-r}{M_1M_2\cdots M_k}\right)^{1-\alpha}\sum_m\mean\sum_{i\in\Gamma_m}Q(i)^\alpha  \\
 & \le\left(\frac{1-r}{M_1M_2\cdots M_k}\right)^{1-\alpha}\sum_m p_{\emptyset}\sum_{i\in\support\ket{\varphi}}Q(i)^\alpha  \\
 & =\left(\frac{1-r}{M_1M_2\cdots M_k}\right)^{1-\alpha}\sum_{i\in\support\ket{\varphi}}Q(i)^\alpha  \\
 & = (1-r)^{1-\alpha}2^{(1-\alpha)\left(\entropy_\alpha(Q)-\sum_{j=1}^k\log M_j\right)}.
\end{split}
\end{equation}

To bound the second term, we use that $x^\alpha\le x$ (when $x\in[0,1]$) to get
\begin{equation}\label{eq:secondtermA}
\begin{split}
\mean\left(1-\sum_m\sum_{i\in\Gamma_m}Q(i)\right)^\alpha r^{1-\alpha}
 & \le r^{1-\alpha}\mean\left(1-\sum_m\sum_{i\in\Gamma_m}Q(i)\right)  \\
 & =r^{1-\alpha}\left(1-\sum_m\mean\sum_{i\in\Gamma_m}Q(i)\right),
\end{split}
\end{equation}
and the lower bound from Lemma~\ref{lem:meanbounds} (with $f(i)=Q(i)$), which implies
\begin{equation}\label{eq:secondtermB}
\begin{split}
\mean\sum_{i\in\Gamma_m}Q(i)
 & \ge \frac{1}{M_1M_2\cdots M_k}\left(1-\sum_{\substack{J\subseteq[k]  \\  J\neq\emptyset,J\neq[k]}}2^{\maxentropy(A_J|A_{\overline{J}})_Q-\sum_{j\in J}\log M_j}\right)\sum_{i\in\support\ket{\varphi}}Q(i)  \\
 & \ge \frac{1}{M_1M_2\cdots M_k}(1-2^{-\Delta}).
\end{split}
\end{equation}
Here we used that $2^k$ is an upper bound on the number of subsets, the definition of $\Delta$ and that $\sum_i Q(i)=1$. Next we combine eqs.~\eqref{eq:secondtermA} and~\eqref{eq:secondtermB}, using that the $m$ takes $M_1M_2\cdots M_k$ different values:
\begin{equation}\label{eq:secondterm}
\mean\left(1-\sum_m\sum_{i\in\Gamma_m}Q(i)\right)^\alpha r^{1-\alpha}\le r^{1-\alpha}2^{-\Delta}.
\end{equation}

By eqs.~\eqref{eq:firstterm} and~\eqref{eq:secondterm}, the expected value of~\eqref{eq:suminlog} is upper bounded by
\begin{equation}
(1-r)^{1-\alpha}2^{(1-\alpha)\left(\entropy_\alpha(Q)-\sum_{j=1}^k\log M_j\right)}+r^{1-\alpha}2^{-\Delta},
\end{equation}
therefore there exists a realization of the random variables achieving a value which does not exceed this bound. According to Lemma~\ref{lem:todiagonals}, there is a probability distribution $P_{XY}$ such that $\ketbra{\varphi}{\varphi}\loccto\GHZ_{P_{XY}}$ and
\begin{equation}
\upentropy[\alpha](X|Y)_P\ge\frac{1}{1-\alpha}\log\left((1-r)^{1-\alpha}2^{(1-\alpha)\left(\entropy_\alpha(Q)-\sum_{j=1}^k\log M_j\right)}+r^{1-\alpha}2^{-\Delta}\right).
\end{equation}
We choose $r$ optimally (as can be seen by differentiation or observing that the right hand side itself is a R\'enyi divergence) as
\begin{equation}
r=\frac{2^{-\Delta/\alpha}}{2^{\frac{1-\alpha}{\alpha}\left(\entropy_\alpha(Q)-\sum_{j=1}^k\log M_j\right)}+2^{-\Delta/\alpha}},
\end{equation}
which leads to the bound
\begin{equation}
\upentropy[\alpha](X|Y)_P\ge\frac{\alpha}{1-\alpha}\log\left(2^{\frac{1-\alpha}{\alpha}\left(\entropy_\alpha(Q)-\sum_{j=1}^k\log M_j\right)}+2^{-\Delta/\alpha}\right)
\end{equation}
for some $P$. Lemma~\ref{lem:approximatefrommixture} implies that $\GHZ_{P_{XY}}\loccto[\epsilon]\GHZ^{\otimes N}$, therefore
\begin{equation}
\ketbra{\varphi}{\varphi}\loccto[\epsilon]\GHZ^{\otimes N}.
\end{equation}
\end{proof}

\begin{corollary}\label{cor:oneshotsimpler}
Using the same notations as in the statement of Theorem~\ref{thm:oneshot}, assume $\Delta>0$ and $\minentropy(Q)>\sum_{j=1}^k\log M_j$. Then for any $\epsilon\in(0,1)$ the relation $\ketbra{\varphi}{\varphi}\loccto[\epsilon]\GHZ^{\otimes N}$ holds with some $N$ satisfying
\begin{equation}\label{eq:simplebound}
N\ge\left(\minentropy(Q)-\sum_{j=1}^k\log M_j\right)\left(1-\frac{1}{\Delta}\right)-\left(2+\frac{\sum_{j=1}^k\log|I_j|}{\Delta}\right)\log\frac{10}{\epsilon^2}.
\end{equation}
\end{corollary}
\begin{proof}
The right hand side of \eqref{eq:oneshotbound} is an increasing function of the entropy of $Q$, so it gets smaller if we replace the entropy by $\minentropy(Q)$. Using the abbreviation $h=\minentropy(Q)-\sum_{j=1}^k\log M_j$, set
\begin{equation}
\alpha=1+\frac{\Delta}{h}.
\end{equation}
Then the logarithm in \eqref{eq:oneshotbound} becomes
\begin{equation}
\log\left(2^{\frac{1-\alpha}{\alpha}h}+2^{-\Delta/\alpha}\right)=1-\frac{\Delta}{\alpha}.
\end{equation}
We get the lower bound
\begin{equation}
\begin{split}
N
 & =\left\lfloor \frac{\alpha}{1-\alpha}\log\left(2^{\frac{1-\alpha}{\alpha}\left(\entropy_\alpha(Q)-\sum_{j=1}^k\log M_j\right)}+2^{-\Delta/\alpha}\right)-\left(1+\frac{1}{\alpha-1}\right)\log\frac{10}{\epsilon^2}\right\rfloor  \\
 & \ge \frac{\alpha}{1-\alpha}\left(1-\frac{\Delta}{\alpha}\right)-\left(1+\frac{1}{\alpha-1}\right)\left(\log\frac{10}{\epsilon^2}\right)-1  \\
 & = h\left(1-\frac{1}{\Delta}\right)-\left(1+\frac{h}{\Delta}\right)\left(\log\frac{10}{\epsilon^2}\right)-2  \\
 & \ge h\left(1-\frac{1}{\Delta}\right)-\left(1+\frac{\sum_{j=1}^k\log|I_j|}{\Delta}\right)\left(\log\frac{10}{\epsilon^2}\right)-2  \\
 & \ge h\left(1-\frac{1}{\Delta}\right)-\left(2+\frac{\sum_{j=1}^k\log|I_j|}{\Delta}\right)\left(\log\frac{10}{\epsilon^2}\right),
\end{split}
\end{equation}
in the last two steps using $h\le\minentropy(Q)\le\sum_{j=1}^k\log|I_j|$ and $\log\frac{10}{\epsilon^2}\ge\log10>2$.
\end{proof}

The asymptotic statement follows by a standard argument involving the asymptotic equipartition property.
\begin{proof}[Proof of Theorem~\ref{thm:asymptotic}]
There is nothing to prove if $\entropy(P)-\sum_{j=1}^kx_j\le 0$. Otherwise, let $\epsilon\in(0,1)$ and $\delta>0$ small enough, and define for every (large) $n\in\mathbb{N}$ the jointly typical set $\mathcal{T}^n$ to be the set of $n$-tuples $(i_1,\ldots,i_n)\in(I_1\times\cdots\times I_k)^n$ such that for every subset $J\subseteq[k]$ the inequality
\begin{equation}
\left|\entropy(A_J)_P-\frac{1}{n}\sum_{m=1}^n\log P_{A_J}((i_{mj})_{j\in J})\right|\le\delta+\frac{1}{n}\log(1-\epsilon)
\end{equation}
holds. Let $\ket{\varphi}\in(\mathbb{C}^{I_1}\otimes\cdots\otimes\mathbb{C}^{I_k})^{\otimes n}$ be the vector with components
\begin{equation}
\varphi_{i_1\ldots i_n}=\begin{cases}
\frac{1}{\sqrt{P^{\otimes n}(T^n)}}\prod_{m=1}^n\psi_{i_n} & \text{if $i_1\ldots i_n\in T^n$}  \\
0 & \text{otherwise.}
\end{cases}
\end{equation}
By the asymptotic equipartition property (see e.g. \cite[Theorem 7.6.1]{cover2012elements})
\begin{equation}
\lim_{n\to\infty}\purifieddistance(\ketbra{\psi^{\otimes n}}{\psi^{\otimes n}},\ketbra{\varphi}{\varphi})=0.
\end{equation}
In particular, for sufficiently large $n$ the distance is less than $\epsilon$. By its definition, the following inequalities hold:
\begin{align}
\forall J\subseteq[k]:\maxentropy(A_J|A_{\overline{J}})_{\Phi'} & \le n(\entropy(A_J|A_{\overline{J}})_\Psi+\delta)  \\
\minentropy(\Phi') & \ge n(\entropy(\Psi)-\delta).
\end{align}

We use Corollary~\ref{cor:oneshotsimpler} with the vector $\ket{\varphi}$ (so that $Q(i)=|\varphi_i|^2$) and $M_j=\lceil 2^{n(x_j+2\delta)+k}\rceil$ (the role of the term $k$ is to cancel the $-k$ in the definition of $\Delta$).
Then for any $J\subseteq[k]$ such that $J\neq\emptyset$, $J\neq[k]$
\begin{equation}
\begin{split}
\sum_{j\in J}\log M_j-\maxentropy(A_J|A_{\overline{J}})-k
 & \ge \sum_{j\in J}\left(n(x_j+2\delta)+k\right)-\maxentropy(A_J|A_{\overline{J}})_Q-k  \\
 & \ge 2n\delta+n\sum_{j\in J}x_j-n(\entropy(A_J|A_{\overline{J}})_P+\delta)  \\
 & \ge n\delta,
\end{split}
\end{equation}
therefore $\Delta\ge n\delta$. On the other hand,
\begin{equation}
\begin{split}
\minentropy(Q)-\sum_{j=1}^k\log M_j
 & \ge n(\entropy(P)-\delta)-\sum_{j=1}^k\left(n(x_j+2\delta)+k\right)-k  \\
 & = n\left(\entropy(P)-\sum_{j=1}^kx_j-(2k+1)\delta\right)-k(k+1).
\end{split}
\end{equation}
Using \eqref{eq:simplebound} we conclude that $\ketbra{\psi}{\psi}^{\otimes n}\loccto[2\epsilon]\GHZ^{\otimes N}$ where
\begin{multline}
N\ge \left(n\left(\entropy(P)-\sum_{j=1}^kx_j-(2k+1)\delta\right)-k(k+1)\right)(1-\frac{1}{n\delta})  \\  -\left(2+\frac{\sum_{j=1}^k\log|I_j|}{\delta}\right)\log\frac{10}{\epsilon^2}.
\end{multline}
Letting $n\to\infty$ this gives the lower bound
\begin{equation}
\liminf_{n\to\infty}\frac{N}{n}\ge\entropy(P)-\sum_{j=1}^kx_j-(2k+1)\delta.
\end{equation}
Since this holds for any $\epsilon\in(0,1)$, we also get
\begin{equation}
\distillablerate(\ketbra{\psi}{\psi})\ge\entropy(P)-\sum_{j=1}^kx_j-(2k+1)\delta.
\end{equation}
Finally, the inequality is true for any $\delta>0$, therefore
\begin{equation}
\distillablerate(\ketbra{\psi}{\psi})\ge\entropy(P)-\sum_{j=1}^kx_j.
\end{equation}
\end{proof}

\begin{remark}
The one-shot result can be formulated in such a way that the resulting state is a pure GHZ state with high probability and the parties have access to a classical flag telling whether the transformation succeeded. That is, the protocol implements an exact LOCC transformation to $(1-o(1))\GHZ^{\otimes N}$ with essentially the same $N$ as before. To see this, one needs an exact probabilistic version of Lemma~\ref{lem:approximatefrommixture}, which also follows from Nielsen's theorem. However, in the asymptotic result the transformation becomes approximate because the projection onto the jointly typical subspace cannot be implemented by LOCC.
\end{remark}

\section{Evaluation of the optimal bound}\label{sec:linearprogram}

To get the best possible bound in Theorem~\ref{thm:asymptotic} or in Example~\ref{ex:uniformrandom}, one needs to minimize the sum $x_1+\cdots+x_k$ subject to the respective constraints. Introducing $h_J=\entropy(A_J|A_{\overline J})$ (respectively using the definition in Lemma~\ref{lem:maxentropyrate}), one is lead to the following linear program:
\begin{problem}[Primal]\label{problem:primal}
Minimize $\displaystyle\sum_{j=1}^k x_j$ subject to $\displaystyle\forall J\subseteq[k],J\neq\emptyset,J\neq[k]:\sum_{j\in J}x_j\ge h_J$.
\end{problem}
This linear program is clearly feasible and bounded, e.g. $x_j=\max_Jh_J$ satisfies the constraints and for every partition $\{J_1,\ldots,J_r\}$ of $[k]$ every feasible point satisfies $\sum_{j=1}^k x_j\ge\sum_{i=1}^r h_{J_i}$. Problem~\ref{problem:primal} is an example of a covering linear program. The dual problem is the following (packing) problem:
\begin{problem}[Dual]\label{problem:dual}
Maximize $\displaystyle\sum_{\substack{J\subseteq[k]  \\  J\neq\emptyset,J\neq[k]}}h_Jy_J$ subject to $\displaystyle\forall j:\sum_{J\ni j}y_J\le 1$, $\forall J:y_J\ge 0$.
\end{problem}

By strong duality, the optimal values of both programs are equal to each other. The advantage of the dual formulation is that the feasible region does not depend on the parameters $h_J$, therefore in principle one can find every vertex for a given $k$ and write down the optimal value as the maximum of a finite number of linear combinations of the parameters. Note that $h_J\ge 0$ implies that the maximum is attained at a vertex satisfying $\forall j:\sum_{J\ni j}y_J=1$. However, the number of such vertices is still very large except for small values of $k$.

To get the bound on the distillable GHZ rate (respectively the asymptotic subrank), one subtracts the optimal value of either linear program from $\entropy(P)$ (respectively $h_{[k]}$).

\begin{example}[Asymptotic subrank, $k=2$]
For $k=2$ the optimum is clearly at $x_j=h_{\{j\}}$ for $j=1,2$ in the primal formulation, $y_{\{1\}}=y_{\{2\}}=1$ in the dual one, which leads to the lower bound
\begin{equation}
\log\asymptoticsubrank(\Psi)\ge h_{\{1,2\}}-h_{\{1\}}-h_{\{2\}}.
\end{equation}
Recall that the terms on the right hand side depend on a chosen distribution $P\in\distributions(\Psi)$ through its marginals. The terms are
\begin{align}
h_{\{1,2\}} & = \max_{\substack{Q\in\distributions(\Psi)  \\  \forall j:Q_j=P_j}}\entropy(A_{\{1,2\}})_Q=\max_{\substack{Q\in\distributions(\Psi)  \\  \forall j:Q_j=P_j}}\entropy(Q)  \\
h_{\{1\}} & = \max_{\substack{Q\in\distributions(\Psi)  \\  \forall j:Q_j=P_j}}\entropy(A_{\{1\}}|A_{\{2\}})_Q=\max_{\substack{Q\in\distributions(\Psi)  \\  \forall j:Q_j=P_j}}\entropy(Q)-\entropy(P_2)  \\
h_{\{2\}} & = \max_{\substack{Q\in\distributions(\Psi)  \\  \forall j:Q_j=P_j}}\entropy(A_{\{2\}}|A_{\{1\}})_Q=\max_{\substack{Q\in\distributions(\Psi)  \\  \forall j:Q_j=P_j}}\entropy(Q)-\entropy(P_1),
\end{align}
therefore
\begin{equation}
\log\asymptoticsubrank(\Psi)\ge\max_{P\in\distributions(\Psi)}\left(\entropy(P_1)+\entropy(P_2)-\max_{\substack{Q\in\distributions(\Psi)  \\  \forall j:Q_j=P_j}}\entropy(Q)\right).
\end{equation}
Without loss of generality the maximum can be restricted to be over those $P$ which have maximal entropy given its marginals, in which case the expression to be maximized is $\mutualinformation(A_1:A_2)_P$.
\end{example}

\begin{example}[Asymptotic subrank, $k=3$]
Let $k=3$. Then the objective function is
\begin{equation}
h_{\{1\}}y_{\{1\}}+h_{\{2\}}y_{\{2\}}+h_{\{3\}}y_{\{3\}}+h_{\{1,2\}}y_{\{1,2\}}+h_{\{1,3\}}y_{\{1,3\}}+h_{\{2,3\}}y_{\{2,3\}},
\end{equation}
and the feasible region is given by the inequalities
\begin{subequations}
\begin{align}
y_{\{1\}}+y_{\{1,2\}}+y_{\{1,3\}} & \le 1  \label{eq:packingconstraintA} \\
y_{\{2\}}+y_{\{1,2\}}+y_{\{2,3\}} & \le 1  \label{eq:packingconstraintB} \\
y_{\{3\}}+y_{\{1,3\}}+y_{\{2,3\}} & \le 1  \label{eq:packingconstraintC} 
\end{align}
\end{subequations}
along with $y_J\ge 0$ for all $J$. The vertices a satisfying eqs.~\eqref{eq:packingconstraintA}--\eqref{eq:packingconstraintC} with equality are $(y_{\{1\}},y_{\{2\}},y_{\{3\}},y_{\{1,2\}},y_{\{1,3\}},y_{\{2,3\}})=(1,1,1,0,0,0)$, $(1,0,0,0,0,1)$, $(0,0,0,1/2,1/2,1/2)$, up to permutation of the subsystem indices. Therefore the maximum of the dual program is
\begin{equation}
\max\Big\{h_{\{1\}}+h_{\{2\}}+h_{\{3\}},h_{\{1\}}+h_{\{1,2\}},h_{\{2\}}+h_{\{1,3\}},h_{\{3\}}+h_{\{1,2\}},  \\  \frac{1}{2}\left(h_{\{1,2\}}+h_{\{1,3\}}+h_{\{2,3\}}\right)\Big\}.
\end{equation}
This has to be subtracted from $h_{\{1,2,3\}}$, so the best lower bound on $\log\asymptoticsubrank(\Psi)$ is the minimum of the following expressions
\begin{subequations}
\begin{gather}
h_{\{1,2,3\}}-h_{\{1\}}-h_{\{2\}}-h_{\{3\}}  \\
h_{\{1,2,3\}}-h_{\{1\}}-h_{\{1,2\}}  \\
h_{\{1,2,3\}}-h_{\{2\}}-h_{\{1,3\}}  \\
h_{\{1,2,3\}}-h_{\{3\}}-h_{\{1,2\}}  \\
h_{\{1,2,3\}}-\frac{1}{2}\left(h_{\{1,2\}}+h_{\{1,3\}}+h_{\{2,3\}}\right).
\end{gather}
\end{subequations}
\end{example}

\begin{example}[Asymptotic subrank of free subsets, $k=3$]
Let $k=3$ again, and assume now that $\Psi$ is free (introduced in ref. \cite{franz2002moment}, this notion of freeness differs from the one in Definition~\ref{def:subrank}), i.e. any two elements of $\Psi$ differ in at least two coordinates. Then $h_{\{j\}}=0$ for $j=1,2,3$, and
\begin{equation}
h_{\overline{j}}=\max_{\substack{Q\in\distributions(\Psi)  \\  \forall j:Q_j=P_j}}\entropy(Q)-\entropy(P_j).
\end{equation}
Assume without loss of generality that $P$ maximizes $\entropy(P)$ with given marginals. Then the lower bound on $\log\asymptoticsubrank(\Psi)$ becomes
\begin{equation}
\log\asymptoticsubrank(\Psi)
  \ge \min\left\{\entropy(P_1),\entropy(P_2),\entropy(P_3),\frac{1}{2}\left(\entropy(P_1)+\entropy(P_2)+\entropy(P_3)-\entropy(P)\right)\right\}.
\end{equation}
\end{example}

\begin{example}[Asymptotic subrank]
More generally, if $k$ is arbitrary and any two elements of $\Psi$ differ in at least $k-1$ coordinates, then $h_J=0$ for $|J|\le k-2$,
\begin{equation}
h_{\overline{j}}=\max_{\substack{Q\in\distributions(\Psi)  \\  \forall j:Q_j=P_j}}\entropy(Q)-\entropy(P_j),
\end{equation}
and
\begin{equation}
h_{[k]}=\max_{\substack{Q\in\distributions(\Psi)  \\  \forall j:Q_j=P_j}}\entropy(Q).
\end{equation}
Assume that $P$ has maximal entropy among the distributions with the same marginals, then $h_{\overline{j}}=\entropy(P)-\entropy(P_j)$ and $h_{[k]}=\entropy(P)$. Now most constraints of Problem~\ref{problem:primal} are vacuous, the only conditions are
\begin{equation}
\forall j:\left(\sum_{j=1}^kx_j\right)-x_j\ge \entropy(P)-\entropy(P_j).
\end{equation}
The optimum becomes
\begin{equation}
\max\left\{\entropy(P)-\entropy(P_1),\ldots,\entropy(P)-\entropy(P_k),\frac{1}{k-1}\left(k\entropy(P)-\entropy(P_1)-\cdots-\entropy(P_k)\right)\right\},
\end{equation}
which leads to the bound
\begin{equation}
\log\asymptoticsubrank(\Psi)\ge\min\left\{\entropy(P_1),\ldots,\entropy(P_k),\frac{1}{k-1}\left(\entropy(P_1)+\cdots+\entropy(P_k)-\entropy(P)\right)\right\}
\end{equation}
\end{example}

We now turn to optimizing the lower bound of Theorem~\ref{thm:asymptotic}. When evaluating the optimum of Problem~\ref{problem:dual} in this setting, the number of vertices that need to be considered can be reduced using strong subadditivity of the entropy. Indeed, suppose that a tentative maximum point satisfies $y_{J_1}>0$, $y_{J_2}>0$ for some nonempty disjoint subsets $J_1,J_2\subseteq[k]$ such that $J_1\cup J_2\neq[k]$. Then the transformation $y_{J_1}\to y_{J_1}-\delta$, $y_{J_2}\to y_{J_2}-\delta$, $y_{J_1\cup J_2}\to y_{J_1\cup J_2}+\delta$ leads to a feasible point if $0<\delta\le\min\{y_{J_1},y_{J_2}\}$, and changes the value of the objective function by
\begin{multline}
\delta(\entropy(A_{J_1\cup J_2}|A_{\overline{J_1\cup J_2}})-\entropy(A_{J_1}|A_{\overline{J_1}})-\entropy(A_{J_2}|A_{\overline{J_2}}))  \\  =\delta(\entropy(A_{\overline{J_1}})+\entropy(A_{\overline{J_2}})-\entropy(A_{\overline{J_1\cup J_2}})-\entropy(A_{[k]})),
\end{multline}
which is nonnegative by the strong subadditivity inequality applied to the random variables $A_{J_1},A_{J_2},A_{\overline{J_1\cup J_2}}$. In particular, any vertex where such a reduction is possible can be excluded when looking for the maximum.

\begin{example}[Distillable entanglement, $k=3$]
For $k=3$, the remaining vertices are $(y_{\{1\}},y_{\{2\}},y_{\{3\}},y_{\{1,2\}},y_{\{1,3\}},y_{\{2,3\}})=(1,0,0,0,0,1)$ and $(0,0,0,1/2,1/2,1/2)$, up to permutations of the subsystems. The lower bound on $\distillablerate$ is
\begin{equation}
\min\{\mutualinformation(A_1:A_2A_3),\mutualinformation(A_2:A_1A_3),\mutualinformation(A_3:A_1A_2),\frac{1}{2}\mutualinformation(A_1:A_2:A_3)\},
\end{equation}
where $\mutualinformation(A_1:A_2:A_3)=\entropy(A_1)+\entropy(A_2)+\entropy(A_3)-\entropy(A_1A_2A_3)$.
\end{example}

\begin{example}[Distillable entanglement, $k=4$]
For $k=4$, the lower bound on $\distillablerate$ is the minimum of the quantities
\begin{subequations}
\begin{gather}
\mutualinformation(A_1:A_2A_3A_4)  \\
\mutualinformation(A_1A_2:A_3A_4)  \\
\frac{1}{2}\mutualinformation(A_1A_2:A_3:A_4)  \\
\frac{1}{3}\mutualinformation(A_1:A_2:A_3:A_4)  \\
\frac{1}{3}\entropy(A_1A_2)+\frac{1}{3}\entropy(A_1A_3)+\frac{1}{3}\entropy(A_2A_3)+\frac{2}{3}\entropy(A_4)-\frac{2}{3}\entropy(A_1A_2A_3A_4)
\end{gather}
\end{subequations}
and the similar ones with permuted subsystem indices.
\end{example}

\begin{remark}
If $I_1=\cdots=I_k$ and the state $\ketbra{\psi}{\psi}$ (or at least the induced distribution $P$) is symmetric, then the optimal value is attained at $x_1=\cdots=x_k$. This is because permutations of any feasible vector $x_j$ are still feasible and the objective function takes the same value on them. In this case, the number of inequalities is only $k-1$, since $\entropy(A_J|A_{\overline{J}})_P$ depends only on $|J|$. The optimal value for $x$ can be written as
\begin{equation}
x=\max_{1\le j\le k-1}\frac{\entropy(A_{[j]}|A_{\overline{[j]}})_P}{j}=\max_{1\le j\le k-1}\frac{\entropy(P)-\entropy(A_{[k-j]})_P}{j},
\end{equation}
whereas the lower bound on the distillable entanglement becomes
\begin{equation}
\distillablerate(\ketbra{\psi}{\psi})\ge\min_{1\le j\le k-1}\frac{k\entropy(A_{[k-j]})_P-(k-j)\entropy(P)}{j}.
\end{equation}
\end{remark}

\begin{example}[W state]
Let $\ket{\psi}=\ket{W_k}=\frac{1}{\sqrt{k}}(\ket{100\ldots 00}+\ket{010\ldots 00}+\cdots+\ket{00\ldots 01})$. As remarked above, the optimal value is attained at some $x_1=\cdots=x_k=x$. The distribution $P$ is now uniform on the support, $\entropy(P)=\log k$ and
\begin{equation}
\entropy(A_{[k-j]})_P=-\frac{k-j}{k}\log\frac{1}{k}-\frac{j}{k}\log\frac{j}{k},
\end{equation}
which leads to the bound
\begin{equation}
\distillablerate(\ketbra{W_k}{W_k})\ge\min_{1\le j\le k-1}\log\frac{k}{j}=\log\frac{k}{k-1}=O(k^{-1}).
\end{equation}

The support of $W_k$ in the computational basis has the property that any probability distribution on it is uniquely determined by its marginals, making the limits in Lemma~\ref{lem:maxentropyrate} especially simple to evaluate. They are equal to the appropriate Shannon conditional entropies, which implies that the lower bound of Example~\ref{ex:uniformrandom} is the same as the lower bound on the LOCC distillable rate. It is known that none of these lower bounds are optimal. The logarithm of the asymptotic subrank of the support is $h(1/k)$ where $h(p)=-p\log p-(1-p)\log(1-p)$ (see \cite{coppersmith1990matrix} for the lower bound ($k=3$) and \cite{vrana2015asymptotic,christandl2016asymptotic} (general $k$) and \cite{strassen1991degeneration,vrana2015asymptotic} for the upper bound). The LOCC distillation rate is not known, but for $k=3$ ref. \cite{smolin2005entanglement} proves a lower bound of $0.64327\ldots$, whereas our lower bound is only $0.58496\ldots$.
\end{example}

\begin{example}[Equal superposition of permutations]
Consider now the state $\ket{\psi}=\frac{1}{\sqrt{k!}}\sum_{\sigma\in S_k}\ket{\sigma(1)\sigma(2)\ldots\sigma(k)}$. Then $\entropy(A_{[k-j]})_P=\log\frac{k!}{j!}$, therefore we get the bound
\begin{equation}
\begin{split}
\distillablerate(\ketbra{\psi}{\psi})
 & \ge\log k!-k\max_{1\le j\le k-1}\frac{\log j!}{j}=\log k!-k\frac{\log(k-1)!}{k-1}  \\
 & =\frac{1}{\ln 2}-\frac{\log k}{2k}+O(k^{-1}).
\end{split}
\end{equation}
\end{example}

Finally, we would like to stress that the bound of Theorem~\ref{thm:asymptotic} depends on the local bases as well. At present we are not aware of an efficient way to optimize the basis choice, and we do not know whether the optimized bound is additive on copies of the same state or considering powers can lead to an improvement. On the other hand, if we fix a basis choice for a single copy and use the tensor power bases, then the bound is additive, since in the dual formulation the feasible region stays the same, whereas the objective function as well as $\entropy(P)$ are multiplied by the number of copies.

\section{Discussion}\label{sec:conclusion}

To help evaluate the strengths and weaknesses of our bound, we compare it with the lower bounds from refs. \cite{smolin2005entanglement} and \cite{streltsov2017rates} on specific families of tripartite states. Let $\rho_{ABC}=\rho=\ketbra{\psi}{\psi}$ denote the initial state, $\rho_A=\Tr_{BC}\rho$, etc. its marginals.

The method suggested in \cite[Example 11]{smolin2005entanglement} uses a protocol simultaneously distilling GHZ states and EPR pairs between a specified pair of parties. If we do this for two pairs of parties, then the resulting EPR pairs can be turned into GHZ states by teleportation. The asymptotic rates of GHZ states and EPR pairs between e.g. parties $A$ and $B$ are $\min\{\entropy(\rho_A),\entropy(\rho_B)\}-\entanglementcost(\rho_{AB})$ and $\entanglementcost(\rho_{AB})$, respectively, where $\entanglementcost$ stands for the entanglement cost, and can be replaced with a higher value as long as both rates stay nonnegative. Suppose that for a fraction $t$ of the initial states we apply the protocol to produce EPR pairs between $A$ and $B$, while for the remaining fraction we produce EPR pairs between $B$ and $C$. The above strategy leads to the following lower bound on $\distillablerate$:
\begin{multline}\label{eq:SVWparametricbound}
t(\min\{\entropy(\rho_A),\entropy(\rho_B)\}-\entanglementcost(\rho_{AB}))  \\  +(1-t)(\min\{\entropy(\rho_A),\entropy(\rho_C)\}-\entanglementcost(\rho_{AC}))+\min\{t\entanglementcost(\rho_{AB}),(1-t)\entanglementcost(\rho_{AC})\}
\end{multline}
This holds for any $t\in[0,1]$ and for any permutation of the three parties. This expression as a function of $t$ is either affine or obtained by gluing together two affine parts. Therefore the maximum is attained either at $t=0$ or $t=1$ or at the point where the arguments of the last minimum coincide, $t=\entanglementcost(\rho_{AC})/(\entanglementcost(\rho_{AB})+\entanglementcost(\rho_{AC}))$. For these values \eqref{eq:SVWparametricbound} evaluates to
\begin{gather}\label{eq:SVWbound}
\min\{\entropy(\rho_A),\entropy(\rho_C)\}-\entanglementcost(\rho_{AC})  \\
\min\{\entropy(\rho_A),\entropy(\rho_B)\}-\entanglementcost(\rho_{AB})  \\
\intertext{and}
\frac{\min\{\entropy(\rho_A),\entropy(\rho_B)\}\entanglementcost(\rho_{AC})+\min\{\entropy(\rho_A),\entropy(\rho_C)\}\entanglementcost(\rho_{AB})-\entanglementcost(\rho_{AB})\entanglementcost(\rho_{AC})}{\entanglementcost(\rho_{AB})+\entanglementcost(\rho_{AC})}
\end{gather}
respectively. The value of $\entanglementcost$ is in general not known, therefore in the graphs below we use the entanglement of formation ($\entanglementofformation$) instead as an upper bound, which for two qubits can be evaluated using Wootters' formula \cite{wootters1998entanglement}.

A different lower bound on the distillable rate comes from \cite[Theorem 2]{streltsov2017rates}, specialized to the GHZ state as the target:
\begin{equation}\label{eq:SMEbound}
\distillablerate(\rho)\ge\min\left\{\frac{\entropy(\rho_A)}{2},\entropy(\rho_B),\entropy(\rho_C)\right\}.
\end{equation}
Again, the same holds for any permutation of the subsystems. To get the best bound, the party with the highest local entropy should take the place of $A$.

For reference, we also compute the upper bound given by the bipartite cuts, namely $\min\{\entropy(\rho_A),\entropy(\rho_B),\entropy(\rho_C)\}$.

As a first simple example let us examine the generalized GHZ states $\GHZ_P\in\states(\mathbb{C}^{\mathcal{X}}\otimes\mathbb{C}^{\mathcal{X}}\otimes\mathbb{C}^{\mathcal{X}})$ where $P\in\distributions(\mathcal{X})$. With respect to the local bases $\{\ket{x}\}_{x\in\mathcal{X}}$, Theorem~\ref{thm:asymptotic} gives $\entropy(P)$, matching the bipartite upper bound. \eqref{eq:SVWbound} evaluates to the same lower bound, while \eqref{eq:SMEbound} gives only half of this rate. For more than three parties, Theorem~\ref{thm:asymptotic} still gives $\entropy(P)$, while the method of ref. \cite{streltsov2017rates} leads to a rate of $\entropy(P)/(k-1)$.

Next we consider the asymmetric W states $\ket{\psi_p}=\sqrt{p}\ket{100}+\sqrt{p}\ket{010}+\sqrt{1-2p}\ket{001}$ where $p\in[0,1/2]$. These interpolate between a separable state $(p=0)$ and an EPR pair shared between $A$ and $B$ ($p=1/2$), both of which have zero distillable GHZ entanglement, while the symmetric W state is recovered by choosing $p=1/3$. We choose the $\{\ket{0},\ket{1}\}$ as the local bases. The lower bound of Theorem~\ref{thm:asymptotic} gives the minimum of $h(1-2p)$ and $h(p)-p$, \eqref{eq:SVWbound} yields
\begin{equation}
\frac{h(p)h\left(\frac{1-\sqrt{1-4p(1-2p)}}{2}\right)+\min\{h(p),h(1-2p)\}h\left(\frac{1-\sqrt{1-4p^2}}{2}\right)}{h\left(\frac{1-\sqrt{1-4p(1-2p)}}{2}\right)+h\left(\frac{1-\sqrt{1-4p^2}}{2}\right)}.
\end{equation}
The lower bound from \eqref{eq:SMEbound} is $\min\{h(p),h(1-2p),\max\{h(p)/2,h(1-2p)/2\}\}$. The three lower bounds together with the bipartite upper bound $\min\{h(p),h(1-2p)\}$ are illustrated in Figure~\ref{fig:Wstate}. When $p$ is sufficiently close to $1/2$, both \eqref{eq:SMEbound} and our bound match the upper bound $h(1-2p)$. More precisely, if we denote by $p^*$ the smallest value in $[1/3,1/2]$ such that $p^*\le p$ implies $\distillablerate(\ketbra{\psi_p}{\psi_p})=h(1-2p)$, then \eqref{eq:SMEbound} shows that $p^*\le0.45569\ldots$ (the solution of $h(p)/2=h(1-2p)$), whereas our bound improves this to $p^*\le0.4359\ldots$ (the solution of $h(p)-p=h(1-2p)$).
\begin{figure}
\centering
\includegraphics{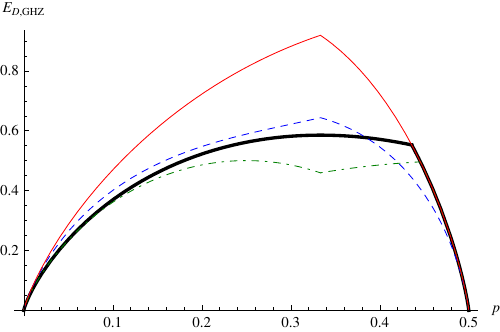}
\caption{\label{fig:Wstate}Bounds on the distillable entanglement of the asymmetric $W$ state. (Thick black: lower bound from Theorem~\ref{thm:asymptotic}, dashed blue: lower bound from \cite{smolin2005entanglement}, dot-dashed green: lower bound from \cite{streltsov2017rates}, solid red: upper bound given by smallest local entropy)}
\end{figure}

Consider now the state
\begin{equation}\label{eq:Rohrlich}
\ket{R_p}=\sqrt{\frac{p}{2}}\ket{000}+\sqrt{\frac{p}{2}}\ket{011}+\sqrt{\frac{1-p}{2}}\ket{100}-\sqrt{\frac{1-p}{2}}\ket{111}.
\end{equation}
This family interpolates between an EPR pair between $B$ and $C$ ($p=0$) and a GHZ state ($p=1/2$, up to an Hadamard gate applied at $A$). This state is studied in \cite[Example 10.]{smolin2005entanglement} (also in \cite{groisman2005entanglement}), where it is found that their protocol gives a GHZ rate of $h(p)$, matching the bipartite upper bound. In contrast, \eqref{eq:SMEbound} gives only $\min\{h(p),1/2\}$, which exemplifies that the method based on combing works best when the state is close to a bipartite one and cannot make use of genuine multipartite entanglement. When applied with the computational basis, Theorem~\ref{thm:asymptotic} only gives the trivial lower bound $0$. This is because the measurement of $A$ is independent of the measurement results of $B$ and $C$ together. However, a rotation on the first qubit can result in a nontrivial bound. Measuring in the $(\ket{0}\pm\ket{1})/\sqrt{2}$ basis leads to a lower bound of $1-h(1/2+\sqrt{p(1-p)})$. The bounds are illustrated in Figure~\ref{fig:Rohrlich}.
\begin{figure}
\centering
\includegraphics{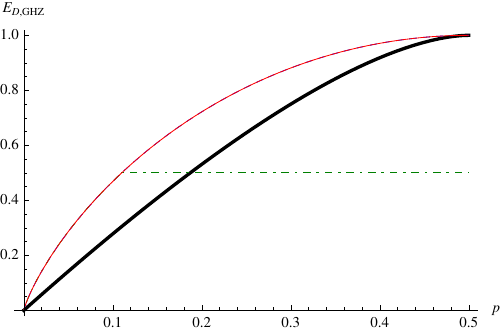}
\caption{\label{fig:Rohrlich}Bounds on the distillable entanglement of the state in \eqref{eq:Rohrlich}. (Thick black: lower bound from Theorem~\ref{thm:asymptotic}, dot-dashed green: lower bound from \cite{streltsov2017rates}, solid red: upper bound given by smallest local entropy, equal to the lower bound from \cite{smolin2005entanglement})}
\end{figure}

\section*{Acknowledgements}

We acknowledge financial support from the European Research Council (ERC Grant Agreement no. 337603) and VILLUM FONDEN via the QMATH Centre of Excellence (Grant no. 10059). This research was supported by the National Research, Development and Innovation Fund of Hungary within the Quantum Technology National Excellence Program (Project Nr.~2017-1.2.1-NKP-2017-00001) and via the research grants K124152, KH~129601 (PV).

\appendix

\section{Entropies}\label{sec:entropies}

In this section we will collect some definitions and facts from single-shot information theory. We mostly follow the notations of ref. \cite{tomamichel2015quantum}, but specialize to classical (commuting) random variables, therefore there is no need to distinguish different types of R\'enyi divergences.

With the help of the R\'enyi divergence
\begin{equation}
\relativeentropy[\alpha]{P}{Q}=\frac{1}{\alpha-1}\log\sum_{x\in\mathcal{X}}P(x)^\alpha Q(x)^{1-\alpha}
\end{equation}
between subnormalized distributions $P,Q\in\subdistributions(\mathcal{X})$, one defines two versions of the R\'enyi conditional entropies (the first one already encountered in the main text, while the second one only used in the appendices with $\alpha\to\infty$),
\begin{equation}
\begin{split}
\upentropy[\alpha](X|Y)_P
 & = \sup_{Q\in\distributions(\mathcal{Y})}-\relativeentropy[\alpha]{P_{XY}}{I_X\otimes Q_Y}  \\
 & = \frac{\alpha}{1-\alpha}\log\left(\sum_{y\in\mathcal{Y}}P_Y(y)\left(\sum_{x\in\mathcal{X}}P_{X|Y}(x|y)^\alpha\right)^{1/\alpha}\right)
\end{split}
\end{equation}
and
\begin{equation}
\begin{split}
\downentropy[\alpha](X|Y)_P
 & = -\relativeentropy[\alpha]{P_{XY}}{I_X\otimes P_Y}  \\
 & = \frac{1}{1-\alpha}\log\left(\sum_{\substack{x\in\mathcal{X}  \\  y\in\mathcal{Y}}}P_{XY}(x,y)^\alpha P_Y(y)^{1-\alpha}\right).
\end{split}
\end{equation}

Both definitions can be extended to $\alpha=\infty$ by taking limits:
\begin{align}
\upentropy[\infty](X|Y)_P & = -\log\sum_{y\in\mathcal{Y}}P_Y(y)\max_{x\in\mathcal{X}}P_{X|Y}(x|y) \\
\downentropy[\infty](X|Y)_P & =\min_{y\in\support P_Y}\min_{x\in\mathcal{X}}\log\frac{1}{P_{X|Y}(x|y)}.
\end{align}
We define the smooth min-entropy as
\begin{equation}
\minentropy[\epsilon](X|Y)_P=\max_{Q\in\ball{\epsilon}{P}}\upentropy[\infty](X|Y)_Q.
\end{equation}
Note that if $P$ is normalized and we allow embedding into a larger alphabet, then the maximum is attained at a normalized distribution $Q$ \cite[Lemma 6.5]{tomamichel2015quantum}. The smooth min-entropy can be lower bounded using the R\'enyi entropies as \cite[eq. (6.92)]{tomamichel2015quantum}
\begin{equation}\label{eq:minentropyfromRenyientropy}
\minentropy[\epsilon](X|Y)_P\ge\upentropy[\alpha](X|Y)_P-\frac{1}{\alpha-1}\log\frac{2}{\epsilon^2}
\end{equation}
for $P\in\subdistributions(\mathcal{X}\times\mathcal{Y})$ and any $\alpha>1$ and $\epsilon\in(0,1)$.

The alternative smooth min-entropy is \cite[Definition 4.]{tomamichel2011leftover}
\begin{equation}
\altminentropy[\epsilon](X|Y)_P=\max_{Q\in\ball{\epsilon}{P}}\downentropy[\infty](X|Y)_Q.
\end{equation}
Note again that if $P$ is normalized then the optimal $Q$ can be chosen normalized as well. The min-entropy and the alternative min-entropy are related as \cite[Lemma 20.]{tomamichel2011leftover}
\begin{equation}\label{eq:altminentropyfromminentropy}
\altminentropy[\epsilon_1+\epsilon_2](X|Y)_P\ge\minentropy[\epsilon_2](X|Y)_P-\log\left(\frac{2}{\epsilon_1^2}+\frac{1}{1-\epsilon_2}\right).
\end{equation}

We now combine the inequalities to bound the alternative smooth min-entropy in terms of the ``up'' R\'enyi entropy.
\begin{lemma}\label{lem:altminentropyfromRenyientropy}
For $P\in\distributions(\mathcal{X}\times\mathcal{Y})$, $\alpha>1$ and $\epsilon\in(0,1)$ the inequality
\begin{equation}
\altminentropy[\epsilon](X|Y)_P\ge\upentropy[\alpha](X|Y)_P-\left(1+\frac{1}{\alpha-1}\right)\log\frac{10}{\epsilon^2}
\end{equation}
holds.
\end{lemma}
\begin{proof}
We use \eqref{eq:altminentropyfromminentropy} with $\epsilon_1=\epsilon_2=\epsilon/2$ and then \eqref{eq:minentropyfromRenyientropy} with smoothing parameter $\epsilon/2$:
\begin{equation}
\begin{split}
\altminentropy[\epsilon](X|Y)_P
 & \ge \minentropy[\epsilon/2](X|Y)_P-\log\left(\frac{8}{\epsilon^2}+\frac{1}{1-\epsilon/2}\right)  \\
 & \ge \upentropy[\alpha](X|Y)_P-\frac{1}{\alpha-1}\log\frac{8}{\epsilon^2}-\log\left(\frac{8}{\epsilon^2}+\frac{1}{1-\epsilon/2}\right)  \\
 & \ge \upentropy[\alpha](X|Y)_P-\frac{1}{\alpha-1}\log\frac{10}{\epsilon^2}-\log\frac{10}{\epsilon^2}
\end{split}
\end{equation}
since $(1-\epsilon/2)^{-1}<2<2/\epsilon^2$.
\end{proof}

\begin{remark}
The bound derived in Lemma~\ref{lem:altminentropyfromRenyientropy} has the advantage that it has a simple form and is valid for the largest possible range of parameters. This comes at the cost of not being tight in certain regimes, including the one we use in the proof of Theorem~\ref{thm:asymptotic} (namely $\alpha\approx 1$, $\epsilon\ll 1$). However, this does not affect the resulting asymptotic bound.
\end{remark}

\section{Properties of random generalized GHZ states}\label{sec:GHZmixtures}

The goal of this section is to derive some properties of random GHZ states (see Definition~\ref{def:randomGHZ}) and to prove Lemma~\ref{lem:approximatefrommixture}. First note that Nielsen's theorem \cite{nielsen1999conditions} extends to a characterisation of LOCC transformations between the pure states $\GHZ_P$:
\begin{theorem}[Nielsen]
Let $P$ and $Q$ be probability distributions. Then $\GHZ_P\loccto\GHZ_Q$ iff $P$ is majorized by $Q$.
\end{theorem}

\begin{lemma}\label{lem:exactfrommixture}
Let $P_{XY}$ be a probability distribution on $\mathcal{X}\times\mathcal{Y}$. Then
\begin{equation}
\GHZ_{P_{XY}}\loccto\GHZ^{\otimes\lfloor\downentropy[\infty](X|Y)_{P}\rfloor}.
\end{equation}
\end{lemma}
\begin{proof}
When the conditioning is trivial ($|\mathcal{Y}|=1$), the relation is an immediate consequence of Nielsen's theorem. In the general case, this implies
\begin{equation}
\begin{split}
\GHZ_{P_{XY}(\cdot|y)}
 & \loccto\GHZ^{\otimes\lfloor\downentropy[\infty](X|Y=y)\rfloor}  \\
 & \loccto\GHZ^{\otimes\lfloor\min_{y\in\support P_Y}\downentropy[\infty](X|Y=y)\rfloor} =\GHZ^{\otimes\lfloor\downentropy[\infty](X|Y)\rfloor}.
\end{split}
\end{equation}
In the second transformation some of the copies are discarded, which is clearly possible via LOCC. To implement the stated transformation, the corresponding transformations are preformed conditioned on the classical value $y$, and then the labels are erased.
\end{proof}

Next we turn to approximate transformations. The following lemma relates the purified distance ($\purifieddistance$) of two random GHZ states to that of the probability distributions.
\begin{lemma}\label{lem:mixturedistance}
Let $P,Q\in\subdistributions(\mathcal{X}\times\mathcal{Y})$. Then $\purifieddistance(\GHZ_{P_{XY}},\GHZ_{Q_{XY}})=\purifieddistance(P,Q)$.
\end{lemma}
\begin{proof}
It is enough to estabilish $\fidelity(\GHZ_{P_{XY}},\GHZ_{Q_{XY}})=\fidelity(P,Q)$, since the purified distance is a function of the fidelity. First note that $\Tr\GHZ_{P_{XY}}=\sum_{x,y}P_{XY}(x,y)$, therefore the first terms in \eqref{eq:fidelity} agree. The second term is homogeneous in both states, therefore we can assume that $P,Q$ are normalized.

The states are block-diagonal with respect to the direct sum decomposition $\mathbb{C}^{\mathcal{X}}\times\mathbb{C}^{\mathcal{Y}}=\bigoplus_{y\in\mathcal{Y}}\mathbb{C}^{\mathcal{X}}$, therefore $\sqrt{\GHZ_{Q_{XY}}^{1/2}\GHZ_{P_{XY}}\GHZ_{Q_{XY}}^{1/2}}$ can be computed blockwise. The contribution of block $y$ to the trace is
\begin{equation}
\sqrt{P_Y(y)}\sqrt{Q_Y(y)}\sum_{x\in\mathcal{X}}\sqrt{P_{X|Y}(x|y)Q_{X|Y}(x|y)}=\sum_{x\in\mathcal{X}}\sqrt{P_{XY}(x,y)Q_{XY}(x,y)}.
\end{equation}
The sum of this expression over $y$ is equal to $\fidelity(P,Q)$.
\end{proof}

\begin{proof}[Proof of Lemma~\ref{lem:approximatefrommixture}]
We may suppose that $|\mathcal{X}|$ is so large that the alternative smooth min-entropy of $P$ is attained at a normalized distribution $P'$. This means that $\purifieddistance(P,P')\le\epsilon$ and (by Lemma~\ref{lem:altminentropyfromRenyientropy})
\begin{equation}
\downentropy[\infty](X|Y)_{P'}=\altminentropy[\epsilon](X|Y)_P\ge\upentropy[\alpha](X|Y)_P-\left(1+\frac{1}{\alpha-1}\right)\log\frac{10}{\epsilon^2}.
\end{equation}
By Lemma~\ref{lem:mixturedistance} we have $\purifieddistance(\GHZ_{P_{XY}},\GHZ_{P'_{XY}})=\purifieddistance(P_{XY},P'_{XY})\le\epsilon$ and by Lemma~\ref{lem:exactfrommixture} the relation $\GHZ_{P'_{XY}}\loccto\GHZ^{\otimes N}$ is true. Therefore $\GHZ_{P_{XY}}\loccto[\epsilon]\GHZ^{\otimes N}$ as claimed.
\end{proof}

\newcommand{\etalchar}[1]{$^{#1}$}

\end{document}